\documentclass[hidelinks,11pt]{article}

\usepackage{amsmath}
\usepackage{amssymb}
\usepackage{amsfonts,amsthm}
\usepackage{thmtools,thm-restate}
\usepackage{tabularx,lipsum,environ}
\usepackage{multirow}
\usepackage{hyperref}

\usepackage{tikz}
\usepackage{thmtools,thm-restate}
\usepackage{tabularx,lipsum,environ}
\usepackage{multirow}

\usepackage[a4paper,bindingoffset=0.2in,
            top=1.1in, bottom=1.1in, left=1.0in, right=1.0in,
            footskip=.5in]{geometry}

\newtheorem{theorem}{Theorem}
\newtheorem{lemma}{Lemma}

\newtheorem{observation}{Observation}

\newcommand{\ignore}[1]{}

\newcommand{\Oh}{\mathcal{O}}

\newcommand{\W}[1]{{\sf W[#1]}}

\let\emptyset\varnothing

\title{Subset Feedback Vertex Set on Graphs of Bounded Independent Set Size\thanks{The second author has been supported by the Hellenic Foundation for Research \& Innovation.}}

\author{
Charis Papadopoulos\footnote{Department of Mathematics, University of Ioannina, Greece. Email: \texttt{charis@cs.uoi.gr}}
\and
Spyridon Tzimas\footnote{Department of Mathematics, University of Ioannina, Greece. Email: \texttt{roytzimas@hotmail.com}.}
}

\date{}

\begin{document}

\maketitle

\begin{abstract}
The (\textsc{Weighted}) \textsc{Subset Feedback Vertex Set} problem is a generalization of the classical \textsc{Feedback Vertex Set} problem and asks for a vertex set of minimum (weighted) size
that intersects all cycles containing a vertex of a predescribed set of vertices.
Although the two problems exhibit different computational complexity on split graphs, no similar characterization is known on other classes of graphs.
Towards the understanding of the complexity difference between the two problems, it is natural to study the importance of a structural graph parameter.  
Here we consider graphs of bounded independent set number for which it is known that \textsc{Weighted Feedback Vertex Set} is solved in polynomial time.
We provide a dichotomy result with respect to the size of a maximum independent set.
In particular we show that \textsc{Weighted Subset Feedback Vertex Set} can be solved in polynomial time for graphs of independent set number at most three,
whereas we prove that the problem remains NP-hard for graphs of independent set number four.
Moreover, we show that the (unweighted) \textsc{Subset Feedback Vertex Set} problem can be solved in polynomial time on graphs of bounded independent set number by giving an algorithm with running time $n^{\Oh(d)}$, where $d$ is the size of a maximum independent set of the input graph.
To complement our results, we demonstrate how our ideas can be extended to other terminal set problems on graphs of bounded independent set size.
Based on our findings for \textsc{Subset Feedback Vertex Set}, we settle the complexity of \textsc{Node Multiway Cut}, a terminal set problem that asks for a vertex set of minimum size that intersects all paths connecting any two terminals, as well as its variants where nodes are weighted and/or the terminals are deletable, for every value of the given independent set number.
\end{abstract}

\section{Introduction}
Given a (vertex-weighted) graph $G=(V,E)$ and a set $S \subseteq V$, the (\textsc{Weighted}) \textsc{Subset Feedback Vertex Set} problem asks for a vertex set of minimum (weighted) size that intersects all cycles containing a vertex of $S$.
It was introduced by Even et al.~who obtained a constant factor approximation algorithm for its weighted version~\cite{EvenNZ00}.
Interestingly \textsc{Subset Feedback Vertex Set} for $|S| = 1$ also coincides with the NP-complete \textsc{Multiway Cut} problem \cite{FominHKPV14} in which the task is to disconnect a predescribed set of vertices \cite{Calinescu08,GargVY04}.
Cygan et al. \cite{CyganPPW13} and Kawarabayashi and Kobayashi \cite{KK12} independently
showed that \textsc{Subset Feedback Vertex Set} is fixed-parameter tractable (FPT) parameterized by the solution size, while Hols and Kratsch provided a randomized polynomial kernel for the problem \cite{HolsK18}.
There has been a considerable amount of work to obtain faster, still exponential-time, algorithms even when restricted to particular graph classes \cite{ChitnisFLMRS13,FominHKPV14,FominGLS16,GolovachHKS14}.

As a generalization of the classical \textsc{Feedback Vertex Set} for which $S=V$, the problem remains NP-hard on bipartite graphs \cite{Yannakakis81a} and planar graphs~\cite{GJ}.
On the positive side, \textsc{Weighted Subset Feedback Vertex Set} can be solved in polynomial time on interval graphs, permutation graphs, and cobipartite graphs \cite{PapadopoulosT17}, the latter being a subclass of graphs of independent set size at most two.
However a notable difference between the two problems regarding their complexity status is the class of split graphs: \textsc{Feedback Vertex Set} is known to be polynomial-time solvable on split graphs \cite{fvs:chord:corneil:1988,Spinrad03},
whereas \textsc{Subset Feedback Vertex Set} remains NP-hard on split graphs \cite{FominHKPV14}.
This gives evidence to the fact that the \textsc{Subset Feedback Vertex Set} seems more difficult to attack than the classical setting of the problem.
Thus it is interesting to explore and obtain further (in)tractability results for \textsc{Subset Feedback Vertex Set}.

Towards such a direction it is reasonable to consider structural parameters of graphs that may lend themselves to provide a unified approach.
In terms of parameterized complexity \textsc{Feedback Vertex Set} is known to be FPT, when parameterized by tree-width \cite{CyganNPPRW11} and clique-width \cite{Bui-XuanSTV13} which implies that \textsc{Feedback Vertex Set} can be solved in polynomial time on graphs of bounded such parameters.
Although \textsc{Feedback Vertex Set} is W[1]-hard parameterized by the size of the independent set\footnote{In Section~\ref{sec:unweightedSFVS} we give a different and simpler reduction from the \textsc{Multicolored Independent Set} problem.}, it can be solved in polynomial time on graphs of bounded maximum induced matching (i.e., \textsc{Feedback Vertex Set} belongs in XP parameterized by the size of the maximum induced matching) \cite{BRV14}.
Only very recently, Jaffke et al. proposed an algorithm that solves \textsc{Weighted Feedback Vertex Set} in time $n^{\Oh(w)}$ where $w$ is the {\it maximum induced matching width} of the given graph \cite{JKTT18}.
Despite their relevant name, graphs of bounded maximum induced matching are not related to graphs of bounded maximum induced matching width as indicated in \cite{Vatschelle2012}.

The approach of \cite{JKTT18} provides a powerful mechanism, as it unifies polynomial-time algorithms for \textsc{Weighted Feedback Vertex Set} on several graph classes such as interval graphs, permutation graphs, circular-arc graphs, and Dilworth-$k$ graphs among others.
Such a mechanism raises the question of whether the algorithm given in \cite{JKTT18} can be extended to the more general setting of \textsc{Weighted Subset Feedback Vertex Set}.
However the proposed algorithm is based on the crucial fact that the forest of a solution has bounded number of internal nodes which is not necessarily true for the $S$-forest of \textsc{Weighted Subset Feedback Vertex Set}.
Thus it seems difficult to control the size of the solution whenever $S \subset V$.
As this observation does not rule out any positive answer, here we develop the first step towards such an approach by considering graphs of bounded independent set number which form candidate relevant graphs.
Notice that graphs of bounded independent set number are not related to graphs of bounded maximum induced matching width.
However graphs of bounded independent set number form the first natural class of bounded structural parameter that are interesting to explore regarding the complexity of \textsc{Subset Feedback Vertex Set}.
Although \textsc{Weighted Feedback Vertex Set} can be solved in time $n^{\Oh(p)}$ on graphs of maximum induced matching at most $p$ \cite{BRV14}, \textsc{Subset Feedback Vertex Set} is already NP-complete on graphs of maximum induced matching equal to one (i.e., split graphs) \cite{FominHKPV14}.

In this work we show that the complexity behaviour of the weighted version of the problem is completely different from the
behaviour of the unweighted variant on graphs with bounded $\alpha(G)$, where $\alpha(G)$ is the size of a maximum independent set in a graph $G$.
\begin{itemize}
\item We show that \textsc{Weighted Subset Feedback Vertex Set} can be solved in polynomial time on graphs with $\alpha(G)\leq 3$.
\end{itemize}
Such graphs consist of the complements of triangle-free graphs; recall that for triangle-free graphs \textsc{Feedback Vertex Set} remains NP-hard \cite{Yannakakis81a}.
Even on such graphs \textsc{Weighted Subset Feedback Vertex Set} requires a structural characterization of the solution with respect to the vertices that are {\it close} to $S$.
\begin{itemize}
\item We further provide a dichotomy result showing that \textsc{Weighted Subset Feedback Vertex Set} remains NP-complete on graphs with $\alpha(G)=4$.
\end{itemize}
Thus we enlarge our knowledge on the complexity difference of the two problems with respect to a structural graph parameter.
\begin{itemize}
\item In order to complement our results we show that \textsc{Subset Feedback Vertex Set} can be solved in time $n^{\Oh(d)}$, where $\alpha(G)\leq d$.
\end{itemize}
Thus we provide a complexity difference between the weighted and the unweighted versions of the problem with respect to a natural structural parameter.
Our main findings concerning \textsc{Subset Feedback Vertex Set} are summarized in Table~\ref{tab:results}.

Moreover, we demonstrate how our ideas can be extended to other {\em terminal set} problems on graphs of bounded independent set size.
In these type of problems we are given a graph $G =(V, E)$, a terminal set $T \subseteq V$, and a nonnegative integer $k$ and the
goal is to find a set $X \subseteq V$ with $|X|\leq k$ which intersects all ``structures'' (such as cycles or paths) passing through the vertices in $T$ \cite{ChitnisFLMRS17}.
The (unweighted) \textsc{Node Multiway Cut} problem is concerned with finding a set $X \subseteq V\setminus T$ of size at most $k$ such that any path between two different terminals intersects $X$.
\textsc{Node Multiway Cut} is known to be in FPT parameterized by the solution size \cite{CLL09,Marx06} and even above guaranteed value \cite{CyganPPW}.
For further results on variants of \textsc{Node Multiway Cut} we refer to \cite{Calinescu08,GargVY04,KratschW12}.
We completely characterize the complexity of \textsc{Node Multiway Cut} with respect to the size of the maximum independent set.
\begin{itemize}
\item In particular, we show that for $\alpha(G)\leq 2$ \textsc{Node Multiway Cut} can be solved in polynomial time,
whereas for $\alpha(G)=3$ it remains NP-complete by adopting the reduction for \textsc{Weighted Subset Feedback Vertex Set} with $\alpha(G)=4$.
\end{itemize}
We further consider a relaxed variation of \textsc{Node Multiway Cut} in which we are allowed to remove terminal vertices,
called \textsc{Node Multiway Cut with Deletable Terminals} (also known as \textsc{Unrestricted Node Multiway Cut}).
\begin{itemize}
\item We show that the (unweighted) \textsc{Node Multiway Cut with Deletable Terminals} problem can be solved in polynomial time on graphs of bounded independent set number,
using an idea similar to the polynomial-time algorithm for the \textsc{Subset Feedback Vertex Set} problem.
\item
We also consider its node-weighted variation and provide a dichotomy complexity result showing that \textsc{Weighted Node Multiway Cut with Deletable Terminals}
can be solved in on graphs with $\alpha(G)\leq 2$, whereas it becomes NP-complete on graphs with $\alpha(G)=3$.
\end{itemize}
It should be noted that the polynomial-time algorithm for the weighted variation is obtained by invoking our algorithm for \textsc{Weighted Subset Feedback Vertex Set}
on graphs with $\alpha(G)\leq 3$. 


\begin{table}[t]
\begin{center}
\setlength\extrarowheight{3pt}
\setlength{\tabcolsep}{4pt}
\begin{tabular}{c | l c c | c|}\cline{2-5}
\multirow{2}{*}{} & \multicolumn{4}{c|}{Bounded Structural Parameter} \\\cline{2-5}
                         & \multicolumn{3}{c|}{Max. Independent Set ($d$)}  & Max. Induced Matching ($p$)\\\hline
\multicolumn{1}{|c|}{Weighted FVS}             & \multicolumn{3}{c}{} & {$n^{\Oh(p)}$ \cite{BRV14}} \\\hline  
\multicolumn{1}{|c|}{\multirow{2}{*}{Weighted SFVS}} &  $d \leq 3$ & $n^{\Oh(1)}$ & Theorem \ref{theo:wsfvspoly} &   \\\cline{2-4} 
\multicolumn{1}{|c|}{}                         &  $d = 4$ & NP-complete & Theorem \ref{theo:wsfvsdNP} &  \\\cline{1-4}
\multicolumn{1}{|c|}{Unweighted SFVS}          &  & {$n^{\Oh(d)}$} & Theorem \ref{theo:sfvspoly} & {NP-complete \cite{FominHKPV14}}\\\hline
\end{tabular}
\end{center}
\caption{Computational complexity results for \textsc{Feedback Vertex Set} (FVS) and \textsc{Subset Feedback Vertex Set} (SFVS) on graphs of bounded independent set number and graphs of bounded maximum induced matching. Note that every graph of independent set number $d$ has maximum induced matching of size at most $d$, while the converse is not necessarily true. }\label{tab:results}
\end{table}

\section{Preliminaries}
We refer to \cite{graph:classes:brandstadt:1999,Diestel12,graph:classes:Go04} for our standard graph terminology.
For $X \subseteq V$, $N_G(X)=\bigcup_{v \in X} N_G(v) \setminus X$ and $N_G[X] = N_G(X) \cup X$.
A {\it weighted graph} $G = (V, E)$ is a graph, where each vertex
$v \in V$ is assigned a {\it weight} that is a positive integer number.
We denote by $w(v)$ the weight of each vertex $v \in V$.
For a vertex set $A \subset V$, the weight of $A$, denoted by $w(A)$, is $\sum_{v \in A} w(v)$.

Given a graph $G$, the {\it independent set number}, denoted by $\alpha(G)$, is the size of the maximum independent set in $G$.
In terms of forbidden subgraph characterization, note that $\alpha(G)\leq d$ if and only if $G$ does not contain $(d+1)K_1$ as an induced subgraph.
We say that a graph $G$ has {\it bounded independent set size} if there exists a positive integer $d$ such that $\alpha(G) \leq d$.
The {\it clique cover number} of $G$, denoted by $\kappa(G)$, is the smallest number of cliques needed to partition $V(G)$ into $S_1, \ldots, S_k$ such that $G[S_i]$ is a clique.
A \emph{vertex cover} is a set of vertices such that every edge of $G$ is incident to at least one vertex of the set.
A \emph{matching} is a set of edges having no common endpoint.
An \emph{induced matching}, denoted by $pK_2$, is a matching $M$ of $p$ edges such that $G[V(M)]$ is isomorphic to $pK_2$.
The \emph{maximum induced matching} number, denoted by $p(G)$, is the largest number of edges in any induced matching of $G$.
It is not difficult to see that for any graph $G$, $\kappa(G) \geq \alpha(G) \geq p(G)$ holds.

Here we consider the following problem.
\vspace*{-0.1in}
\tikzstyle{mybox} = [draw=black, fill=none, very thin,
    rectangle, rounded corners, inner sep=7pt, inner ysep=8pt]
\tikzstyle{fancytitle} =[fill=white, text=black]
\begin{figure}[!h]
\centering
\begin{tikzpicture}
\node [mybox] (box){%
  \noindent
  \begin{tabularx}{0.9\textwidth}{@{\hspace{\parindent}} l X c}
    \textit{Input:} & A (vertex-weighted) graph $G$, a set $S \subseteq V$, and a nonnegative integer $k$. \\[1pt]
    \textit{Task:} & Decide whether there is a set $X \subseteq V$ with $|X| \leq k$ ($w(X) \leq k$) such that no cycle in $G-X$ contains a vertex of $S$. 
  \end{tabularx}};
\node[fancytitle, right=6pt] at (box.north west) {(\textsc{Weighted}) \textsc{Subset Feedback Vertex Set} -- SFVS};
\end{tikzpicture}%
\end{figure}
\vspace*{-0.1in}

\noindent As remarked, we distinguish between the weighted and the unweighted version of the problem.
In the unweighted version of the problem note that all weights are equal and positive.
The classical \textsc{Feedback Vertex Set} (FVS) problem is a special
case of \textsc{Subset Feedback Vertex Set} with $S = V$.
A vertex of $S$ is simply called {\it $S$-vertex}.
%
An induced cycle of $G$ is called {\it $S$-cycle} if an $S$-vertex is contained in the cycle.
We define an \emph{$S$-forest} $F=(V_F,E_F)$ to be the subgraph of $G$ induced by the vertex set $V_F \subseteq V$ for which no cycle in $G[V_F]$ is an $S$-cycle.
It is not difficult to see that the problem of computing a minimum weighted subset feedback vertex set is equivalent to the problem of computing a maximum weighted $S$-forest.

Let us give a couple of observations on the nature of \textsc{Subset Feedback Vertex Set} on graphs of bounded independent set size.
Let $G$ be a graph and let $d$ be a positive integer such that every independent set of $G$ has at most $d$ vertices.
Firstly note that the bounded-size independent set is a hereditary property, meaning that for every induced subgraph $H$ of $G$, we have $\alpha(H) \leq d$.
Moreover for any clique $C$ of $G$, any $S$-forest of $G$ contains at most two vertices of $S \cap C$.

\begin{observation}\label{obs:S2d}
Let $G$ be a graph with $\alpha(G)\leq d$ and let $S \subseteq V$.
\begin{itemize}
\item[(1)] For any set $X$ of $2d+1$ vertices, there is a cycle in $G[X]$.
\item[(2)] Any $S$-forest of $G$ has at most $2d$ vertices from $S$.
\end{itemize}
\end{observation}
\begin{proof}
Let $X$ be a set of $2d+1$ vertices.
Assume that $G[X]$ is a forest.
As an induced subgraph of $G$, any independent set of $G[X]$ has size at most $d$.
Since $G[X]$ is acyclic, there is a proper 2-coloring $A,B$ of the vertices of $G[X]$ such that $|A|\geq |B|$.
By the fact that $|A|\leq d$, we conclude that $|A|+|B| \leq 2d$, leading to a contradiction that $|X| \geq 2d+1$.
Thus $G[X]$ contains a cycle.

For the second statement, let $F=(V_F,E_F)$ be an $S$-forest of $G$.
By the first statement, if $F[S]$ has at last $2d+1$ vertices then there is a cycle in $F$ that passes through a vertex of $S$, which implies an $S$-cycle in $F$.
Thus $|S\cap V_F|\leq 2d$.
%
\end{proof}
We note that Observation~\ref{obs:S2d} directly implies that any $2d+1$ vertices of $G[S]$ induce an $S$-cycle, which allows us to construct
by brute force all possible subsets of $S$ belonging to any $S$-forest in time $n^{\Oh(d)}$.

\section{Weighted SFVS on Graphs of Bounded Independent Set}
Here we consider the \textsc{Weighted Subset Feedback Vertex Set} and we show a dichotomy result with respect to the size of the maximum independent set.
We first provide a polynomial-time algorithm on graphs of independent set size at most three and then we show that \textsc{Weighted Subset Feedback Vertex Set} is NP-complete on graphs of independent set size at most four.

Let $(G,S,k,d)$ be an instance of \textsc{Weighted Subset Feedback Vertex Set}
for which $G$ is a graph of independent set size at most $d$.
In the forthcoming arguments, instead of directly computing a solution for \textsc{Weighted Subset Feedback Vertex Set},
we consider the equivalent problem of computing an $S$-forest of $G$ having weight at least $w(V)-k$.

Let $F=(V_F,E_F)$ be an $S$-forest of $G$. 
We partition the $S$-forest $F$ into two induced subgraphs $F_{\leq 1}$ and $F_{>1}$ as follows:
\begin{itemize}
\item $F_{\leq 1}$ is the subgraph of $F$ induced by the vertices of $N[S \cap V_F]$; the vertices of $F_{\leq 1}$ are at distance at most one from $S \cap V_F$ and are denoted by $S_{\leq 1}$.
\item $F_{>1}$ is the graph $F - S_{\leq 1}$ and contains vertices that are at distance at least two from $S \cap V_F$.
\end{itemize}
Such a partition is called {\it $S$-distance partition} of $F$, denoted by $(F_{\leq 1}, F_{>1})$.
The set of edges of $F$ having one endpoint in $F_{\leq 1}$ and the other in $F_{>1}$ are called {\it the cut} with respect to $F_{\leq 1}$ and $F_{>1}$.
Notice that a vertex of $F_{\leq 1}$ that is adjacent to a vertex of $F_{>1}$ belongs to $S_{\leq 1} \setminus S$.

Let $(C_1, \ldots, C_{d'})$ be a partition of the vertices of $F_{>1}$ such that each $C_i$ induces a connected component in $F_{>1}$.
Because $F_{>1}$ is an induced subgraph of $G$, it is clear that $d' \leq d$.
Let $(A_{1},\ldots,A_{d'})$ be a tuple of $d'$ subsets of $S_{\leq 1} \setminus S$, i.e., each $A_i \subseteq (S_{\leq 1} \setminus S)$ holds.
We say that {\it the cut satisfies} the tuple $(A_{1},\ldots,A_{d'})$ if for any vertex $v \in C_i$, we have $(N_G(v) \cap S_{\leq 1}) \subseteq A_i$.
The notion of an $S$-distance partition of $F$ with the corresponding cut is illustrated in Figure~\ref{S-forest-struct}.

\begin{figure}[t]
\centering
\includegraphics[scale= 1.0]{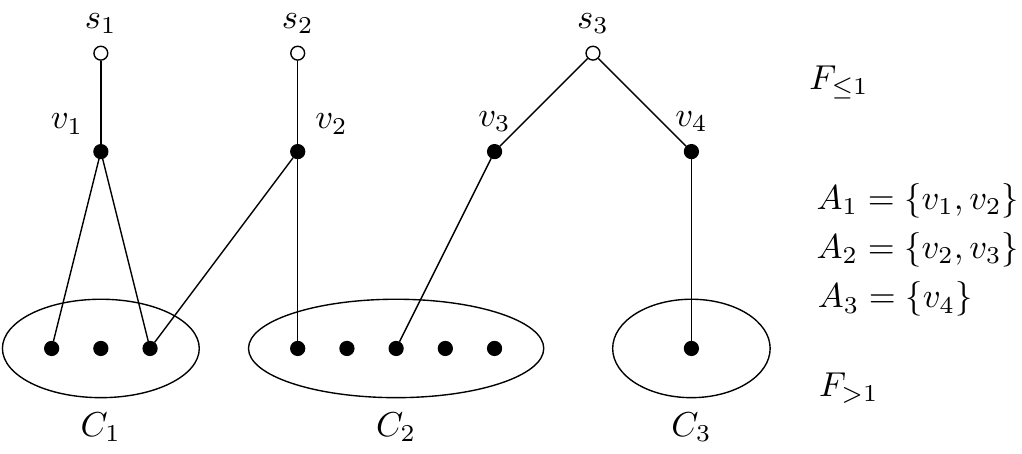}
\caption{Illustrating an $S$-distance partition $(F_{\leq 1}, F_{>1})$ of an $S$-forest $F$ with $S=\{s_1,s_2,s_3\}$ that shows the connected components $C_1, C_2, C_3$ of $F_{>1}$. The edges inside $F_{>1}$ are not drawn in order to highlight that the cut satisfies the given tuple $(A_1,A_2,A_3)$.}
\label{S-forest-struct}
\end{figure}

We now utilize the $S$-distance partition of $F$ in order to construct an algorithm that solves \textsc{Weighted Subset Feedback Vertex Set}
on graphs of independent set size at most $d$ and subsequently show that this algorithm is efficient for $d \leq 3$.
Our general approach relies on the following facts:
\begin{itemize}
\item By Observation~\ref{obs:S2d}~(2) we try all subsets $S'$ of $S$ with at most $2d$ vertices and keep those sets that induce an $S$-forest. 
This step is responsible for constructing the graph $F_{\leq 1}$.
    We will show that the number of the produced such subsets is bounded by $n^{\Oh(d)}$.

\item For each of the potential subsets $S'$ constructed in the previous step, and for each $d'\leq d$, we determine all possible tuples $(A_{1},\ldots,A_{d'})$ in $F_{\leq 1}$ 
with $A_i \subseteq (S' \setminus S)$ that are only satisfied by cuts of $S$-distance partitions of induced subgraphs of $G$ with $S_{\leq 1}=S'$ that are $S$-forests. 
We show why considering only these tuples is sufficient in Lemma \ref{lem:correct}.

\item Up to that point we can show that all steps can be executed in polynomial time regardless of $d\leq 3$.
However for the next and final step we can only achieve polynomial running time if we restrict ourselves to $d\leq 3$ due to the number of connected components of $F_{>1}$.
Then for each tuple selected in the previous step we find connected components $C_{1}, \ldots, C_{d'}$
	of maximum weight such that the cut of $(G[S'],G[C_{1}\cup\cdots\cup C_{d'}])$ satisfies the tuple.
\end{itemize}

We begin by showing that the $S$-distance partition of $F$ provides a useful tool towards computing a maximum $S$-forest.
Given a set of vertices $X \subseteq N[S]$ and $d'$ subsets $A_i$ of $X \setminus S$,
we construct the graph $\widehat{G}$ that is obtained from $G[X]$ by adding $d'$ vertices $w_1, \ldots, w_{d'}$ such that every vertex $w_i$ is adjacent to all the vertices of $A_i$.
In what follows, we always assume that $G$ is a graph having independent set size at most $d$.

\begin{lemma}\label{lem:correct}
Let $F$ be an $S$-forest of $G$ with $S$-distance partition $(F_{\leq 1}, F_{>1})$ such that $S_{\leq 1} \cap S \neq \emptyset$.
Then there is a tuple $(A_1, \ldots, A_{d'})$ with $A_i \subseteq (S_{\leq 1} \setminus S)$ such that
\begin{itemize}
\item[(i)] the cut of $(F_{\leq 1}, F_{>1})$ satisfies $(A_1, \ldots, A_{d'})$ and
\item[(ii)] every induced subgraph $H$ of $G$ with $S$-distance partition $(H[S_{\leq 1}],H-S_{\leq 1})$ that satisfies $(A_1, \ldots, A_{d'})$ is an $S$-forest.
\end{itemize}
\end{lemma}
\begin{proof}
Let $(C_1, \ldots, C_{d'})$ be a partition of the vertices of $F_{>1}$ such that every $C_i$ induces a connected component in $F_{>1}$.
We define a tuple $(A_1, \ldots, A_{d'})$ in which every $A_i = N(C_i) \cap S_{\leq 1}$.
Clearly $A_i \subseteq (S_{\leq 1} \setminus S)$ since every vertex $F_{>1}$ is at distance at least two from $S_{\leq 1} \cap S$.
Thus, by construction, the cut of $(F_{\leq 1}, F_{>1})$ satisfies the tuple $(A_1, \ldots, A_{d'})$.

For the next claim we first show that $\widehat{G}$ with respect to ${S_{\leq 1}}$ and the tuple $(A_1, \ldots, A_{d'})$ is an $S$-forest.
Assume for contradiction that there is an $S$-cycle $\widehat{C}$ in $\widehat{G}$.
Since $F_{\leq 1}$ does not contain any $S$-cycle, $\widehat{C}$ contains a vertex $w_i$ and at least two vertices $u_i, v_i$ from $A_i$.
By the fact that $A_i = N(C_i) \cap S_{\leq 1}$, there is a vertex $x$ in $C_i$ of $F_{>1}$ that is adjacent to $u_i$ and
there is a vertex $y$ in $C_i$ of $F_{>1}$ that is adjacent to $v_i$.
Together with a path between $x$ and $y$ in the connected component $C_i$, we construct a path in $G$ with endvertices $u_i$ and $v_i$ that is completely contained in $C_i$.
This means that if we replace every vertex $w_i$ of $\widehat{G}$ then
we obtain an $S$-cycle in $F$, leading to a contradiction.
Thus $\widehat{G}$ is an $S$-forest.

Let $H$ be an induced subgraph of $G$ with $S$-distance partition $(H[S_{\leq 1}],H-S_{\leq 1})$ that satisfies $(A_1, \ldots, A_{d'})$.
Observe that $H[S_{\leq 1}]=F_{\leq 1}$ as they are induced subgraphs of the same vertex set of $G$.
Thus $H[S_{\leq 1}]$ does not contain any $S$-cycle, because $F$ is an $S$-forest.
Since the cut of $(H[S_{\leq 1}],H-S_{\leq 1})$ satisfies $(A_1, \ldots, A_{d'})$, there is a partition $(T_1, \ldots T_{d'})$ in $H-S_{\leq 1}$ such that
$T_i$ is a connected component of $H-S_{\leq 1}$
and $N(T_i) \subseteq A_i$.
We show that $H$ is indeed an $S$-forest. For contradiction, assume an $S$-cycle $C$ in $H$.
There are no $S$-cycles in $H[S_{\leq 1}]$ which implies that $C \cap T_i \neq \emptyset$ for some $1 \leq i \leq d'$.
For every such set we replace the part $C \cap T_i$ by a vertex $w'_i$.
Denote by $H'$ the resulting graph.
Notice that $H'[C]$ is a subgraph of $\widehat{G}[C]$ because $N_{H'}(w'_i) \subseteq N_{\widehat{G}}(w_i)$.
This, however, implies an $S$-cycle in $\widehat{G}$ which gives the desired contradiction.
Therefore $H$ is an $S$-forest.
\end{proof}

Next we show how to bound the $S_{\leq 1}$ vertices of $F_{\leq 1}$.

\begin{lemma}\label{closed-N-obs}
Let $F$ be an $S$-forest of $G$ with $S$-distance partition $(F_{\leq 1}, F_{>1})$ such that $S_{\leq 1} \cap S \neq \emptyset$.
\begin{enumerate}
\item If $|S_{\leq 1} \cap S| \leq 2d-2$ then $|S_{\leq 1}| \leq 4d-2$.
\item If $|S_{\leq 1} \cap S| \geq 2d-1$ then $|S_{\leq 1}| \leq 2d$.
\end{enumerate}
\end{lemma}
\begin{proof}
Let $F$ be such an $S$-forest of $G$ with $|S_{\leq 1} \cap S|\geq1$.
By Observation~\ref{obs:S2d}~(2), we know that $|S_{\leq 1} \cap S|\leq 2d$.
To ease the presentation, we let $S'=S_{\leq 1} \setminus S$.
We consider separately the two cases of the claim.

\medskip

\noindent \textrm{Case 1.}
Let $1\leq |S_{\leq 1} \cap S| \leq 2d-2$. 
Assume for contradiction that $|S'|>4d-|S_{\leq 1} \cap S|-2$.
We show that $F[S']$ contains a matching with at least $d$ edges.
Observe that $|S'|+|S_{\leq 1} \cap S|>4d-2$.
Applying Observation~\ref{obs:S2d}~(1) shows that there is a cycle $C$ in $F[S_{\leq 1}]$.
Since $F$ is an $S$-forest, this is not an $S$-cycle, so all vertices contained in $C$ are vertices of $S'$.
Iteratively picking the two endpoints of an edge from $C$ as long as $|S'|+|S_{\leq 1} \cap S|> 2d$, constructs $d$ edges of $S'$ having no common endpoints.
Thus $F[S']$ contains a matching $M$ with at least $d$ edges.

Let $C_1, \ldots, C_{d'}$ be the connected components of $F[S_{\leq 1} \cap S]$.
Notice that $d'\leq d$ because $F[S_{\leq 1} \cap S]$ is an induced subgraph of a graph with maximum independent set size at most $d$.
By construction, every vertex of $S'$ is adjacent to at least one vertex of $S_{\leq 1} \cap S$.
If the endpoints of an edge of $M$ in $S'$ are adjacent to vertices of the same component $C_i$ then there is an $S$-cycle in $F$ since every vertex of $C_i$ belongs to $S$.
Thus the endpoints of every edge of $M$ are adjacent to different connected components of $F[S_{\leq 1} \cap S]$.
Now obtain a bipartite graph by contracting every component $C_i$ into a single vertex and every edge of $M$ into a single vertex and keep only the adjacencies between the components and the edges of $M$.
Let $(A,B)$ be the bipartition of the resulting bipartite graph such that $A$ contains the components of $F[S_{\leq 1} \cap S]$ and $B$ contains the edges of $M$.
Since $|A| \leq |B|$ and every vertex of $B$ is adjacent to at least two vertices of $A$, there is a cycle in the bipartite graph.
Then it is not difficult to see that the cycle of the contracted vertices corresponds to an $S$-cycle in $F$.
Therefore there is an $S$-cycle in an $S$-forest, leading to a contradiction.

\medskip

\noindent \textrm{Case 2.}
Let $2d-1 \leq|S_{\leq 1} \cap S| \leq 2d$.
Assume for contradiction that $|S'|>2d-|S_{\leq 1} \cap S|$.
This means that $S'$ contains at least one vertex.
We pick a nonempty subset $W$ of $S'$ as follows.
If $|S_{\leq 1} \cap S| = 2d-1$ then $W$ consists of any two vertices of $S'$.
If $|S_{\leq 1} \cap S| = 2d$ then $W$ consists of an arbitrary vertex of $S'$.
In both cases, notice that $|S_{\leq 1} \cap S|+|W|>2d$ by the fact $2d-1 \leq|S_{\leq 1} \cap S|$.
Then Observation~\ref{obs:S2d}~(1) implies that there is a cycle in $F[(S_{\leq 1} \cap S) \cup W]$.
Since $W$ has at most two vertices, we conclude that the induced cycle of $F[(S_{\leq 1} \cap S) \cup W]$ has at least one vertex from $S$, hence it is an $S$-cycle in $F$.
Therefore we reach a contradiction which implies that $|S'|\leq 2d-|S_{\leq 1} \cap S|$.
\end{proof}

Lemma~\ref{closed-N-obs} shows that we can compute all possible  candidates for $S_{\leq 1}$ in polynomial time as follows.
\begin{itemize}
\item We first try by brute force all subsets $S'$ of $S$ having at most $2d$ vertices, according to Observation~\ref{obs:S2d}~(2).
\item Then, for each such subset $S'$ we incorporate the neighbors $N(S')$ of $S'$ for which either $|N(S')|+|S'| \leq 4d-2$ or $|N(S')|+|S'| \leq 2d$ according to Lemma~\ref{closed-N-obs}.
\item Given the described sets $S'$ and $N(S')$, we check if $G[S'\cup N(S')]$ induces an $S$-forest and, if so, we include them into a list $L_1$ containing all candidates for $S_{\leq 1}$.
\end{itemize}
The correctness follows from Observation~\ref{obs:S2d} and Lemma~\ref{closed-N-obs}.
Regarding the running time notice that we create at most $n^{\Oh(d)}$ subsets for each of $S'$ and $N(S')$.
Thus in time $n^{\Oh(d)}$ we can compute a list $L_1$ that contains all possible subsets of the vertices corresponding to $S_{\leq 1}$.
Notice that such vertices are enough to build the part $F_{\leq 1}$.

Let $S_{\leq 1}$ be a set of $L_1$.
We now focus on the graph $G'=G-(S_{\leq 1} \cup S)$ that contains the vertices that are at distance at least two from $S_{\leq 1} \cap S$.
Notice that all possible vertices in an $S$-forest that are in distance at most one from $S$ are described in $L_1$.
Let $d'$ be the number of connected components of $G'$.
It is clear that $d'\leq d$.
In fact, if $S_{\leq 1}\cap S$ contains at least one vertex then $d' < d$, since the vertices of $G'$ are at distance at least two from $S$.
Moreover, observe that if $S_{\leq 1} \cap S=\emptyset$ then $G-S$ is a trivial solution.
From now on, we assume that $|S_{\leq 1} \cap S|\geq 1$ so that $d' <d$.

By brute force, we find all tuples $(A_1, \ldots, A_{d'})$ such that the following hold: 
\begin{itemize}
\item[(i)] $A_i \subseteq (S_{\leq 1} \setminus S)$ and
\item[(ii)] the graph $\widehat{G}$ with respect to $S_{\leq 1}$ and $(A_1, \ldots, A_{d'})$ is an $S$-forest.
\end{itemize}
Notice that by Lemma~\ref{lem:correct} it is sufficient to consider only such tuples.
Since $A_i \subseteq S_{\leq 1}$, $d'<d$, and $|S_{\leq 1}|\leq 4d$, the number of tuples is $2^{\Oh(d)}$, so that we can obtain the desired set of tuples that satisfy both conditions in polynomial time. 

In what follows, we consider the case for $d\leq 3$.
By the previous arguments we are given a set $S_{\leq 1} \subseteq N[S]$ and tuples of the form $A_1$ or $(A_1,A_2)$ which are subsets of $S_{\leq 1} \setminus S$.
Our task is to compute a subset $V'$ of the vertices of $G'$ such that the vertices of $S_{\leq 1} \cup V'$ induce a maximum $S$-forest and the cut $(G[S_{\leq 1}],G[V'])$ satisfies $A_1$ or $(A_1,A_2)$, respectively.
We distinguish the two cases.

\begin{lemma}\label{lem:caseA1}
Let $X \subseteq N[S]$ and let $A_1$ be a subset of $X \setminus S$ such that both $F_{\leq 1}=G[X]$ and $\widehat{G}$ with respect to $X$ and $A_1$ are $S$-forests. 
There exists a polynomial-time algorithm that computes a maximum $S$-forest $F$ with an $S$-distance partition $(F_{\leq 1}, F_{>1})$ having a cut satisfying $A_1$.
\end{lemma}
\begin{proof}
Since $F_{\leq 1}$ is a fixed $S$-forest of $F$, we need to determine the vertices of $V\setminus(X\cup S)$ that are included in $F_{>1}$.
By the desired cut of $(F_{\leq 1}, F_{>1})$ we are restricted to the vertices of $V\setminus(X\cup S)$ that have neighbors only to $A_1$.
Those vertices can be described as follows:
$$
B_1=(V\setminus (X\cup S))\setminus\left\{w \in V: N(w) \cap (X \setminus (S \cup A_1)) \neq \emptyset\right\}.
$$
Notice that $B_1$ contains vertices that are at distance at least two from the $S$-vertices of $X \cap S$.
Since the cut satisfies a single subset $A_1$, we have at most one connected component of $G[B_1]$ in $F_{> 1}$.
In order to choose the correct connected component of $G[B_1]$, we try to include each of them in $F_{> 1}$ and select the one having the maximum total weight.
Notice that adding any component of $G[B_1]$ into $F_{> 1}$ cannot create any $S$-cycle because $\widehat{G}$ with respect to $X$ and $A_1$ is an $S$-forest.
Thus by Lemma~\ref{lem:correct} we correctly compute a maximum $S$-forest with the desired properties.
Clearly the set $B_1$ can be constructed in polynomial time.
Since the number of connected components $G[B_1]$ is at most two, all steps can be executed in polynomial time.
\end{proof}

Next we consider the tuple $(A_1,A_2)$.

\begin{lemma}\label{lem:caseA1A2}
Let $X \subseteq N[S]$ and let $A_1,A_2$ be subsets of $X \setminus S$ such that both $F_{\leq 1}=G[X]$ and $\widehat{G}$ with respect to $X$ and $(A_1,A_2)$ are $S$-forests.
There exists a polynomial-time algorithm that computes a maximum $S$-forest $F$ with an $S$-distance partition $(F_{\leq 1}, F_{>1})$ having a cut satisfying $(A_1,A_2)$.
\end{lemma}
\begin{proof}
Similar to the proof of Lemma~\ref{lem:caseA1}, we first construct the sets $B_1,B_2$ that contain vertices of $V\setminus (X \cup S)$ and satisfy the cut obtained from $X$:
\begin{align*}
B_1=(V\setminus (X \cup S))\setminus\left\{w \in V: N(w) \cap (X \setminus (S \cup A_1)) \neq \emptyset\right\}\phantom{.} & \textrm{ and } \\
B_2=(V\setminus (X \cup S))\setminus\left\{w \in V: N(w) \cap (X \setminus (S \cup A_2)) \neq \emptyset\right\}.&
\end{align*}
As the desired cut of $(F_{\leq 1}, F_{>1})$ satisfies $(A_1,A_2)$, there are two connected components of $F_{>1}$ which are subsets of these two sets respectively.
Let $C_1$ and $C_2$ be the connected components of $F_{>1}$ such that $C_1 \subseteq B_1$ and $C_2 \subseteq B_2$.
Now observe that there should be two non-adjacent vertices $w_{1}\in B_{1}$ and $w_{2}\in B_{2}$ that belong to $C_1$ and $C_2$, respectively.
We iterate over all possible pairs of non-adjacent vertices $w_{1} \in B_1 \cap C_1$ and $w_{2} \in B_2 \cap C_2$ in $\Oh(n^{2})$ time. 
Assuming a given choice for $w_{1}$ and $w_{2}$, observe the following:
\begin{itemize}
\item Since $w_{1}$ and $w_{2}$ are vertices of different connected components of $F_{>1}$, the components themselves are further restricted to be subsets of $B_{1}\setminus N[w_{2}]$ and $B_{2}\setminus N[w_{1}]$, respectively. That is, $C_1 \subseteq (B_{1}\setminus N[w_{2}])$ and $C_2 \subseteq B_{2}\setminus N[w_{1}]$.
\item Since $F$ has at least one vertex of $S$, $w_{1},w_{2}\in V \setminus (X \cup S)$ are non-adjacent, and by the fact $d\leq 3$, we have that $B_{1}\setminus N[w_{2}]$ and $B_{2}\setminus N[w_{1}]$ induce cliques in $G$. Thus $B_{1}\setminus N[w_{2}] \subseteq N[w_{1}]$ and $B_{2}\setminus N[w_{1}] \subseteq N[w_{2}]$, respectively.
\end{itemize}
Then by the second statement it is not difficult to see that $B_{1}\setminus N[w_{2}]$ and $B_{2}\setminus N[w_{1}]$ are disjoint.
Let $B'_1 = (B_{1}\setminus N[w_{2}])\setminus\{w_{1}\}$ and $B'_2 = (B_{2}\setminus N[w_{1}])\setminus\{w_{2}\}$.
Now in order to find the maximum induced $S$-forest under the stated conditions and our assumption that $w_{1}$ and $w_{2}$ belong to the two connected components of $F_{>1}$ it suffices to find the maximum subset $C_1 \cup C_2$ of $B'_1\cup B'_2$ such that there are no edges between the vertices of $C_1 \cap B'_1$ and the vertices of $C_2 \cap B'_2$.
This boils down to compute a minimum weighted vertex cover on the bipartite graph $G'$ obtained from $G[B'_1\cup B'_2]$ and removing the edges inside $G[B'_1]$ and $G[B'_2]$.
By maximum flow standard techniques, we compute a minimum weighted vertex cover $U$ on $G'$ in polynomial time \cite{Orlin13}.
Therefore $G[B'_1\cup B'_2] - U$ contains the connected components $C_1\setminus \{w_1\}$ and $C_1\setminus \{w_2\}$, as required.
\end{proof}

Now we are equipped with our necessary tools in order to obtain our main result, namely a polynomial-time algorithm that solves \textsc{Weighted Subset Feedback Vertex Set} on graphs of independent set of size at most $3$.

\begin{theorem}\label{theo:wsfvspoly}
\textsc{Weighted Subset Feedback Vertex Set} on graphs of independent set of size at most $3$ can be solved on time $n^{\Oh(1)}$.
\end{theorem}
\begin{proof}
Let us briefly explain such an algorithm for computing a maximum $S$-forest $F$ of a graph $G$ having independent set size at most three. Let $d=3$.
Initially we set $F^{*} = G - S$.
Then for every set $X \subseteq N[S]$ with $|X|\leq 4d$ such that $G[X]$ is an $S$-forest, we try by brute force all subsets $A_1$ and $(A_1,A_2)$ with $A_i \subseteq (X \setminus S)$ such that $\widehat{G}$ with respect to $X$ and $A_1$ or $(A_1,A_2)$ is an $S$-forest.
For each of such subsets we find a maximum $S$-forest $F$ with an $S$-distance partition $(G[X], F_{>1})$ having a cut satisfying $A_1$ or $(A_1,A_2)$, respectively, by applying the algorithms described in Lemma~\ref{lem:caseA1} and Lemma~\ref{lem:caseA1A2}.
At each step, we maintain the maximum weighted $S$-forest $F^*$ by comparing $F$ with $F^*$.
Finally we provide the vertices $V \setminus V(F^*)$ as the set with the minimum total weight that are removed from $G$.

By Lemma~\ref{closed-N-obs}, it is sufficient to consider the described subsets $X$.
Since every induced subgraph of $G-X$ contains at most two connected components,
Lemma~\ref{lem:correct} implies that all possible subsets $A_1$ or $(A_1,A_2)$ with the described properties are enough to consider.
Thus the correctness follows from Lemmata~\ref{closed-N-obs}--\ref{lem:caseA1A2}.
Regarding the running time, notice that whether a graph contains an $S$-cycle can be tested in linear time.
Thus we can construct all described and valid subsets in $n^{\Oh(1)}$ time.
Therefore the total running time of the algorithm takes time $n^{\Oh(1)}$, since each of the algorithms given in Lemma~\ref{lem:caseA1} and Lemma~\ref{lem:caseA1A2}, respectively, requires polynomial time.
\end{proof}


Let us now show that extending Theorem~\ref{theo:wsfvspoly} to graphs of larger independent sets is not possible.
More precisely with the following result we show that \textsc{Weighted Subset Feedback Vertex Set} is {para-NP-complete} parameterized by $\alpha(G)$.

\begin{theorem}\label{theo:wsfvsdNP}
\textsc{Weighted Subset Feedback Vertex Set} is NP-complete on graphs of independent set of size at most $4$.
\end{theorem}
\begin{proof}
We will provide a polynomial reduction from the \textsc{Vertex Cover} (VC) problem on tripartite graphs which is NP-complete \cite{GJ}.
Let $G=(A,B,C,E)$ be a tripartite graph where $(A,B,C)$ is the partition of $V(G)$.
We construct a weighted graph $G'$ from $G$ in polynomial time as follows.
\begin{itemize}
\item We turn the three independent sets $A$, $B$ and $C$ into cliques by adding all necessary edges and we give all vertices unary weight. 
\item We add a vertex $r_{A}$ that is adjacent to all of the vertices of $A$ and we assign weight $n$ to $r_A$.
In a completely symmetric way, we add vertices $r_B$ and $r_C$ with respect to the sets $B$ and $C$, respectively.
\item We add a vertex $s$ that is adjacent to all three vertices $r_A, r_B, r_C$ and we assign weight $n$ to $s$.
\end{itemize}
This completes the construction of $G'$.
Observe that all vertices of $V(G')\setminus \{s,r_A,r_B,r_C\}$ have weight equal to one.
It is not difficult to verify that the constructed graph $G'$ is a graph having an independent set at most $4$, since $G'-\{s\}$ is a vertex-disjoint union of three cliques.

Next we claim that $G$ has a vertex cover $U$ of size at most $k<n$ if and only if $G'$ with $S=\{s\}$ has a subset feedback vertex set of weight at most $k$.
Assume a vertex cover $U$ of $G$.
By definition, $U$ covers all edges of $G$, so that $G[(A\cup B\cup C)\setminus U]$ is an independent set.
This means that $G'[(A\cup B\cup C)\setminus U]$ is a vertex-disjoint union of cliques.
Since $s$ is non-adjacent to any vertex of $G$ and $G'[r_A,r_B,r_C]$ is an independent set, every cycle of $G' - U$ contains a vertex of $r_A,r_B$ and $r_C$ with at least two vertices from $A,B$ and $C$, respectively.
Thus $G' - U$ is a connected $S$-forest.
Therefore $U$ is a subset feedback vertex set of $(G',\{s\})$ of size at most $k$.

For the opposite direction, assume a subset feedback vertex set $F$ of $(G',\{s\})$.
If $F$ is not a subset of $A\cup B\cup C$, then its sum of weights is greater or equal to $n$.
Then $F$ is not a minimum subset feedback vertex set of $(G',\{s\})$, since $A\cup B\cup C$ minus a single vertex is trivially a subset feedback vertex set of $(G',\{s\})$ of total weight $n-1$.
Thus $F$ is indeed a subset of $A\cup B\cup C$.
Assume that $F$ is not a vertex cover of $G$.
By definition, there is an edge of $G$ that remains uncovered.
Without loss of generality, assume that this edge has its endpoints on the vertices $x\in A$ and $y\in B$.
Then $\langle s,r_{A},x,y,r_{B}\rangle$ is an induced cycle of $G'$, which contradicts the fact that $F$ is a subset feedback vertex set of $(G',\{s\})$.
Therefore $F$ is a vertex cover of $G$.
\end{proof}

We stress that Theorem~\ref{theo:wsfvsdNP} further implies that the NP-completeness result carries along to graphs of clique cover number at most four,
since the constructed graph given in the proof can be partitioned into four disjoint cliques.

\section{SFVS on Graphs of Bounded Independent Set}\label{sec:unweightedSFVS}
Here we show that despite the complexity dichotomy result for the \textsc{Weighted Subset Feedback Vertex Set},
whenever the weights of the vertices are equal \textsc{Subset Feedback Vertex Set} can be solved in polynomial time on graphs of bounded independent set number.

\begin{theorem}\label{theo:sfvspoly}
\textsc{Subset Feedback Vertex Set} on graphs of independent set of size at most $d$ can be solved in time $n^{\Oh(d)}$.
\end{theorem}
\begin{proof}
Let $G=(V,E)$ be a graph with $\alpha(G) \leq d$ and let $S \subseteq V$.
Denote by $X\subseteq V$ a minimum subset feedback vertex set of $G$.
Let $F=G-X$ be a maximum $S$-forest of $G$.
By Observation~\ref{obs:S2d}~(2), the vertices of $S$ that belong to $F$ are at most $2d$.
Thus for every optimum solution $X$, the set $S\setminus X$ has at most $2d$ vertices.

Now we claim that it is enough to consider subsets $X'$ of $X$ for which $|X'|\leq 2d$.
To see this, observe that if $X\setminus S$ has order more than $2d$, then $G-S$ has more vertices than $G-X$, leading to a contradiction to the optimality of $X$.
Hence, $X\setminus S$ has at most $2d$ vertices.
In order to find an optimal solution, it suffices to consider all such candidates $S'$ for $S\setminus X$ and $X'$ for $X\setminus S$.
To check whether an induced subgraph of $G$ consists an $S$-forest takes $\Oh(n+m)$ time.
Since the number of such sets $S'$ is at most $n^{2d}$ and the number of the considered sets $X'$ is at most $n^{2d}$, the total running time is bounded by $n^{\Oh(d)}$.
Therefore in time $n^{\Oh(d)}$ we compute a minimum subset feedback vertex set showing the claimed result.
\end{proof}

Regarding the dependence of the exponent in the running time of the algorithm given in Theorem~\ref{theo:sfvspoly},
note that we can hardly avoid this fact, since \textsc{Feedback Vertex Set} is W[1]-hard parameterized by the independent set number as explicitly given in \cite{BRV14}.
At the same time such an observation follows from the W[1]-hardness result from the construction given in \cite{arxivJKT17} with respect to the maximum induced matching width.
In the following result, we provide a different and simpler reduction from the \textsc{Multicolored Independent Set} problem \cite{FellowsHRV09,Pietrzak03} which shows an interesting connection with graphs of bounded independent set size.

\begin{theorem}
\textsc{Feedback Vertex Set} is W[1]-hard when parameterized by the clique cover number.
\end{theorem}
\begin{proof}
The reduction comes from the \textsc{Multicolored Independent Set} problem: given a graph $G$ and a partition $(V_1, \ldots, V_k)$ of $V(G)$, decide whether $G$ contains an independent set of size $k$ using exactly one vertex from each $V_i$.
It is known that \textsc{Multicolored Independent Set} is \W{1}-hard parameterized by $k$ \cite{FellowsHRV09,Pietrzak03}.
Let $(G, V_1, \ldots, V_k)$ be an instance of \textsc{Multicolored Independent Set}.
From $G$ we construct a graph $H$ as follows.
\begin{itemize}
\item We make every set $V_i$ clique by adding all necessary edges.
\item For each $V_i$ we add two vertices $x_i,y_i$ that are adjacent to every vertex of $V_i$.
\item We add a vertex $z$ that is adjacent to all the vertices of $G$.
\end{itemize}
This completes the construction of $H$.
Observe that $|V(H)|=n+2k+1$.
Let $X=\{x_1, \ldots, x_k\}$ and $Y=\{y_1, \ldots, y_k\}$.
Then $X \cup Y \cup \{z\}$ forms an independent set in $H$ of size $2k+1$.
Notice also that the vertices of $V_i \cup \{x_i\}$ induce a clique, so that $H$ has a clique partition of size $2k+1$.
Thus the clique cover number of $H$ is at most $2k+1$ which implies that $H$ has a clique cover number that is linearly dependent on $k$.
We claim that $G$ has a multicolored independent set if and only if $H$ has a feedback vertex set of size at most $n - k$.

Let $I_k$ be $k$ vertices from each of $V_1, \ldots, V_k$ that form an independent set in $G$.
We describe an induced forest $F$ of $H$ that contains $3k+1$ vertices starting from the vertices of $I_k$.
For each vertex $v_i$ of $I_k \cap V_i$ we add in $F$ both vertices $x_i$ and $y_i$.
Notice that $F$ contains $k$ disjoint trees of the form $\{v_i,x_i,y_i\}$.
Since $z$ is non-adjacent to every vertex of $X \cup Y$, we can safely include $z$ in $F$.
Thus $F$ is an induced forest with $3k+1$ vertices, so that $V(H)\setminus V(F)$ constitutes a feedback vertex set of size $n-k$.

For the opposite direction, let $U$ be a feedback vertex set of size $n-k$.
Then $F = H - U$ is an induced forest of $H$ that has at least $3k+1$ vertices.
We claim that from each $V_i$ there is at most one vertex in $F$ and $(X \cup Y) \subseteq V(F)$.
Assume for contradiction that at least two vertices from $V_i$ are contained in $F$.
Since $V_i$ is a clique in $H$, there are exactly two vertices $v_i,v'_i$ from $V_i$ in $F$.
Then neither $x_i$ nor $y_i$ is included in $F$ because they are both adjacent to $v_i$ and $v'_i$.
Then $F' = (F \setminus {v'_i}) \cup \{x_i,y_i\}$ is an induced forest of $H$, as $x_i,y_i$ are non-adjacent to any vertex of $H - V_i$.
This, however, shows that there is an induced forest with at least $|F|+1$ vertices, leading to a contradiction.
Thus $|V_i \cap V(F)| \leq 1$ which implies that $|V(F) \cap V(G)| \leq k$.
Then observe that $X\cup Y$ is an independent set in $H$ and no vertex from $X\cup Y$ induces a cycle with the vertices from $F$, so that $(X \cup Y) \subseteq V(F)$.
If $z \notin V(F)$ then $|V(F)|\leq 3k$.
Hence $z \in V(F)$ and $|V(F)|\leq 3k+1$.
Since $F$ contains at least $3k+1$ vertices, each of $V_i$ contains exactly one vertex in $F$.
Moreover the vertices of $V(G) \cap V(F)$ are pairwise non-adjacent because $z$ is a vertex of $F$ and $z$ is adjacent to every vertex of $G$.
Therefore the $k$ vertices of each of $V_i \cap V(F)$ form an independent set in $G$.
\end{proof}

\section{Extending to other Terminal Set Problems}\label{sec:terminals}
Let us now consider further {\em terminal set} problems that are related to \textsc{Subset Feedback Vertex Set}.
In these type of problems we are given a graph $G =(V, E)$, a terminal set $T \subseteq V$, and a nonnegative integer $k$ and the
goal is to find a set $X \subseteq V$ with $|X|\leq k$ which intersects all ``structures'' (such as cycles or paths) passing through the vertices in $T$ \cite{ChitnisFLMRS17}.
In this setting \textsc{Subset Feedback Vertex Set} is a particular terminal set problem when the objective structure is a cycle.
We show that the ideas that we developed for \textsc{Subset Feedback Vertex Set} on graphs of bounded independent set size, can be extended to further terminal set problems when the objective structure is a path instead of a cycle.

The (unweighted) \textsc{Node Multiway Cut} problem is formulated as follows.

\tikzstyle{mybox} = [draw=black, fill=none, very thin,
    rectangle, rounded corners, inner sep=7pt, inner ysep=8pt]
\tikzstyle{fancytitle} =[fill=white, text=black]
\begin{figure}[!h]
\centering
\begin{tikzpicture}
\node [mybox] (box){%
  \noindent
  \begin{tabularx}{0.9\textwidth}{@{\hspace{\parindent}} l X c}
    \textit{Input:} & A graph $G$, a set $T \subseteq V$ of \textit{terminals}, and a nonnegative integer $k$. \\[1pt]
    \textit{Task:} & Decide whether there is a set $X \subseteq V\setminus T$ with $|X|\leq k$ such that any path between two different terminals intersects $X$. 
  \end{tabularx}};
\node[fancytitle, right=6pt] at (box.north west) {\textsc{Node Multiway Cut}};
\end{tikzpicture}%
\end{figure}

\noindent Notice that in this problem we are not allowed to remove any terminal.
For graphs having independent set size at most $d$ we completely characterize the complexity of \textsc{Node Multiway Cut}.
In particular, for $d=3$ we can adopt the reduction given in Theorem~\ref{theo:wsfvsdNP}.

\begin{theorem}\label{theo:NMCNP}
Let $G$ be a graph of independent set of size at most $d$.
If $d\leq 2$ then \textsc{Node Multiway Cut} can be solved on time $n^{\Oh(1)}$.
Otherwise, \textsc{Node Multiway Cut} is NP-complete on graphs of independent set of size at most $3$.
\end{theorem}
\begin{proof}
Let $(G,T,k)$ be an instance of \textsc{Node Multiway Cut}.
If $G[T]$ contains an edge then we conclude that $(G,T,k)$ is a no-instance, since we are not allowed to remove any vertex from $T$.
In what follows we assume that $G[T]$ is an independent set.
If $d\leq 2$ there are at most two terminals, so that $|T|=2$, and we can solve the problem by standard maximum flow techniques~\cite{Orlin13}.

For $d=3$, we give a reduction from the NP-complete \textsc{Vertex Cover} problem on tripartite graphs, similar to the one given in Theorem~\ref{theo:wsfvsdNP}.
Let $G=(A,B,C,E)$ be a tripartite graph where $(A,B,C)$ is the partition of $V(G)$.
We construct a graph $G'$ from $G$ by making the three independent sets $A$, $B$ and $C$ into cliques and adding three new vertices $r_A, r_B, r_C$, that are adjacent to every vertex of $A$, $B$, and $C$, respectively.
It is clear that $G'$ has independent set size $3$. We let $T=\{r_A, r_B, r_C\}$ and claim that $G$ has a vertex cover $U$ of size at most $k$ if and only if $G'$ has a set $X$ of size at most $k$ which intersects every path between the vertices of $T$.
Removing a vertex cover $U$ from $G$ results in a vertex-disjoint union of three cliques in $G'$ in which each of the vertices $r_A,r_B,r_C$ belongs to a separate clique.
Thus $X=U$ is a solution for \textsc{Node Multiway Cut} on $G'$.
For the opposite direction, observe that $X$ cannot contain any of the three vertices $r_A,r_B,r_C$.
Assume that $X$ is not a vertex cover of $G$.
Then there is an edge $\{a,b\}$ that is not covered by $X$ where $a$ and $b$ belong to different partitions of $V(G)$.
Let $r_a$ and $r_b$ be the terminal vertices of $\{r_A,r_B,r_C\}$ which are adjacent to $a$ and $b$, respectively, in $G'$.
Then it is clear that there is a path between the terminals $r_a$ and $r_b$ in $G' - X$, leading to a contradiction.
Therefore, $X$ is a vertex cover of $G$ of size at most $k$.
\end{proof}

Due to the difficulty of \textsc{Node Multiway Cut} even for the unweighted version and with a small size of independent set, we consider a relaxed variation in which we are allowed to remove terminal vertices.

\vspace*{-0.1in}
\tikzstyle{mybox} = [draw=black, fill=none, very thin,
    rectangle, rounded corners, inner sep=7pt, inner ysep=8pt]
\tikzstyle{fancytitle} =[fill=white, text=black]
\begin{figure}[!h]
\centering
\begin{tikzpicture}
\node [mybox] (box){%
  \noindent
  \begin{tabularx}{0.9\textwidth}{@{\hspace{\parindent}} l X c}
    \textit{Input:} & A (vertex-weighted) graph $G$, a set $T \subseteq V$ of \textit{terminals}, and a nonnegative integer $k$. \\[1pt]
    \textit{Task:} & Decide whether there is a set $X \subseteq V$ with $|X|\leq k$ ($w(X)\leq k$) such that any path between two different terminals intersects $X$. 
  \end{tabularx}};
\node[fancytitle, right=6pt] at (box.north west) {(\textsc{Weighted}) \textsc{Node Multiway Cut with Deletable Terminals}};
\end{tikzpicture}%
\end{figure}

\noindent Next we show that the (unweighted) \textsc{Node Multiway Cut with Deletable Terminals} problem can be solved in polynomial time on graphs of bounded independent set number, using an idea similar to the one given in Theorem~\ref{theo:sfvspoly}.

\begin{theorem}\label{theo:NMCDTP}
\textsc{Node Multiway Cut with Deletable Terminals} on graphs of independent set of size at most $d$ can be solved in time $n^{\Oh(d)}$.
\end{theorem}
\begin{proof}
Let $(G,T,k)$ be an instance of \textsc{Node Multiway Cut with Deletable Terminals} where $G$ is a graph having independent set size at most $d$.
Observe that every solution $X$ has size at most $|T|$.
Assume first that $|T|\leq d$.
Then we can enumerate all subsets having at most $|T|$ vertices in time $n^{\Oh(|T|)}$ and pick the smallest subset that separates all terminals.
Thus in time $n^{\Oh(d)}$ we output a valid solution $X$, if it exists.

Next assume that $d < |T|$.
We consider the graph $G[T]$. As an induced subgraph of $G$, $G[T]$ has independent set size at most $d$.
Thus $G[T]$ contains at least one edge.
If both endpoints of an edge in $G[T]$ do not belong to solution $X$, then there is a path between terminal vertices.
This means that there is a minimum vertex cover $U$ of $G[T]$ such that $U \subseteq X$.
To compute such a set $U$ we enumerate all independent sets $T'\subseteq T$ of size at most $d$ in time $|T|^{\Oh(d)}$ and construct $U = T \setminus T'$.
For each constructed $U$ we consider the graph $G'=G-U$ with terminals $T'$.
Since $T'$ is an independent set in $G'$, we know that $|T'|\leq d$.
Thus in time $n^{\Oh(|T'|)}$ we can compute a set $X'$ of minimum size such that all terminals of $G'-X'$ are separated.
Therefore the total running time is bounded by $|T|^{\Oh(d)}\cdot n^{\Oh(|T'|)}$ which is bounded by $n^{\Oh(d)}$, because $|T|\leq n$ and $|T'| \leq d$, and this gives the claimed running time.
\end{proof}

Let us also stress that we can hardly avoid the dependence of the exponent in the running time given in Theorem~\ref{theo:NMCDTP}.
This comes from the fact that \textsc{Node Multiway Cut with Deletable Terminals} with $T=V(G)$ is equivalent to asking whether
the graph contains a maximum independent set.
That is, we have to solve the \textsc{Independent Set} which is known to be W[1]-hard parameterized by the size of the independent set \cite{DF13}.

Regarding the node-weighted variant of \textsc{Node Multiway Cut with Deletable Terminals}, we can provide a dichotomy result with respect to the size $d$ of a maximum independent set.
In fact, for $d \leq 2$ we can invoke the algorithm for the \textsc{Weighted Subset Feedback Vertex Set} given in Theorem~\ref{theo:wsfvspoly}.
Moreover, due to its close connection to the \textsc{Node Multiway Cut}, for $d>2$ we can assign appropriate weights to the terminals in a way that they become undeletable.
Both ideas are explained in the proof of the following result.

\begin{theorem}\label{theo:weightedNMCDTP}
Let $G$ be a graph of independent set of size at most $d$.
If $d\leq 2$ then \textsc{Weighted Node Multiway Cut with Deletable Terminals} can be solved on time $n^{\Oh(1)}$.
Otherwise, \textsc{Weighted Node Multiway Cut with Deletable Terminals} is NP-complete on graphs of independent set of size at most $3$.
\end{theorem}
\begin{proof}
Let $(G,T,k)$ be an instance of \textsc{Weighted Node Multiway Cut with Deletable Terminals}.
Assume first that $d\leq 2$.
We create an equivalent instance for the \textsc{Weighted Subset Feedback Vertex Set} problem.
Starting from $G$, we obtain a new graph $G'$ by adding a vertex $s$ that is adjacent to all terminals of $T$.
Since we only added one vertex, $G'$ has a maximum independent set of size at most $3$.
We let $S=\{s\}$ and assign a large weight to $s$ that is equal to the sum of the weights of all vertices in $G$.
Next we claim that any solution for the \textsc{Weighted Subset Feedback Vertex Set} on $G'$ consists a valid solution for the \textsc{Weighted Node Multiway Cut with Deletable Terminals}.
Notice that a solution in $G'$ cannot contain the new vertex $s$ due its assigned weight.
Also observe that any cycle in $G'$ passing through $s$ corresponds to a path in $G$ connecting two terminals of $T$.
Thus by running the algorithm of Theorem~\ref{theo:wsfvspoly} on $G'$, we obtain a solution for the \textsc{Weighted Node Multiway Cut with Deletable Terminals} problem in time  $n^{\Oh(1)}$.

Now assume that $d=3$.
Given an instance $(G,T,k)$ for the (unweighted) \textsc{Node Multiway Cut}, we assign weight $n$ to every terminal of $T$ and unary weight to every other vertex.
Thus the solutions for both problems contain only non-terminal vertices which implies that they are equivalent.
Therefore the NP-completeness of \textsc{Weighted Node Multiway Cut with Deletable Terminals} follows, since the (unweighted) \textsc{Node Multiway Cut} is NP-complete on graphs of independent set size at most three by Theorem~\ref{theo:NMCNP}.
\end{proof}

\section{Concluding Remarks}\label{sec:concl}
We conclude with a few open problems.
Despite the fact that the \textsc{Weighted Subset Feedback Vertex Set} is NP-complete on graphs with bounded independent set number,
it is still interesting to settle the complexity of \textsc{Subset Feedback Vertex Set} on graphs of maximum induced matching width by extending the approach given in \cite{JKTT18}.
Towards such a direction, Dilworth-$k$ graphs seem a possible candidate for clarifying the complexity status of \textsc{Subset Feedback Vertex Set} (for an exposition of such parameters, see for e.g. \cite{Vatschelle2012}).
Moreover, \textsc{Feedback Vertex Set} is known to be polynomially-time solvable on cocomparability graphs \cite{fvs:cocomp:liang:1997}, and, more generally, on AT-free graphs \cite{KratschMT08}.
To our knowledge, \textsc{Subset Feedback Vertex Set} has not been studied on such graphs, besides the existence of a fast exponential-time algorithm for the unweighted variant of the problem \cite{ChitnisFLMRS17}.
Concerning such an approach, our results indicate that it is natural and compelling to settle first the unweighted \textsc{Subset Feedback Vertex Set} problem.
Furthermore, Theorem~\ref{theo:NMCNP} shows that \textsc{Node Multiway Cut} remains NP-complete on graphs having maximum induced matching $3$.
However, on graphs of bounded maximum induced matching the complexity of \textsc{Node Multiway Cut with Deletable Terminals} is still unknown.

\bibliography{subsetFVS_classes}

\clearpage
\appendix

\end{document}